\documentclass[twocolumn,amsmath,amssymb, nofootinbib]{revtex4-1}\usepackage{float}
\usepackage{amsmath}
\usepackage{amssymb}
\usepackage{amsthm}
\usepackage{float}
\usepackage{hyperref}
\usepackage{color}
\usepackage{simplewick}

\newtheorem{Theorem}{Theorem}
\theoremstyle{Lemma}
\newtheorem{Lemma}[Theorem]{Lemma}

\DeclareSymbolFont{extraup}{U}{zavm}{m}{n}
\DeclareMathSymbol{\varheart}{\mathalpha}{extraup}{86}
\DeclareMathSymbol{\vardiamond}{\mathalpha}{extraup}{87}

\begin{document}
   \title{Problems with Propagation and Time Evolution in $f(T)$ Gravity}

  \author{Yen Chin Ong}
   \email{ongyenchin@member.ams.org}
    \affiliation{Graduate Institute of Astrophysics, National Taiwan University, Taipei 10617, Taiwan.}
 		\affiliation{Leung Center for Cosmology and Particle Astrophysics, National Taiwan University, Taipei 10617, Taiwan.}

\author{Keisuke Izumi}
   \email{izumi@phys.ntu.edu.tw}
   	\affiliation{Leung Center for Cosmology and Particle Astrophysics, National Taiwan University, Taipei 10617, Taiwan.}
   
 \author{James M. Nester}
   \email{nester@phy.ncu.edu.tw}
   \affiliation{Department of Physics, National Central University, Chungli 320, Taiwan.}
 \affiliation{Graduate Institute of Astronomy, National Central University, Chungli 320, Taiwan.}
 \affiliation{Center for Mathematics and Theoretical Physics, National Central University, Chungli 320, Taiwan.}
 \affiliation{Institute of Physics, Academia Sinica, Taipei 11529, Taiwan.}

 \author{Pisin Chen}
   \email{chen@slac.stanford.edu}
     \affiliation{Graduate Institute of Astrophysics, National Taiwan University, Taipei 10617, Taiwan.}
 		\affiliation{Leung Center for Cosmology and Particle Astrophysics, National Taiwan University, Taipei 10617, Taiwan.}
    \affiliation{Department of Physics, National Taiwan University, Taipei 10617, Taiwan.}
    \affiliation{Kavli Institute for Particle Astrophysics and Cosmology, SLAC National Accelerator Laboratory, Stanford University, Stanford, CA 94305, U.S.A.}

\begin{abstract}Teleparallel theories of gravity have a long history. They include a special case referred to as the \emph{Teleparallel Equivalent of General Relativity} (TEGR, aka GR$_{\|}$).  Recently this theory has been generalized to $f(T)$ gravity. Tight constraints from observations suggest that $f(T)$ gravity is not as robust as initially hoped.  This might hint at hitherto undiscovered problems at the theoretical level. In this work, we point out that a generic $f(T)$ theory can be expected to have certain problems including superluminal propagating modes, the presence of which can be revealed by using the characteristic equations that govern the dynamics in $f(T)$ gravity and/or the Hamiltonian structure of the theory via Dirac constraint analysis. We use several examples from simpler gauge field theories to explain how such superluminal modes could arise.
We also point out problems with the Cauchy development of a constant time hypersurface in FLRW spacetime in $f(T)$ gravity. The time evolution from a FLRW (and as a special case, Minkowski spacetime) initial condition is not unique. 
\end{abstract}

\maketitle
\section{Introduction}

Einstein proposed the idea of teleparallelism, or \textit{Fernparallelismus}  (aka \textit{absolute parallelism}, \textit{distant parallelism}), in order to unify gravity and electromagnetism~\cite{Einstein}. Unlike general relativity, spacetime in teleparallelism is equipped with a connection with zero curvature, but non-vanishing torsion. Since the spacetime is \emph{flat}, the parallel transport of a vector is independent of path, and hence the name \textit{tele}parallel, meaning ``parallel at a distance''. Einstein's quest for unification via teleparallel theory can be found in the interesting account
of~\cite{Sauer}. Despite Einstein's failure to formulate a unified theory, teleparallelism was later revived and studied as a pure theory of gravity. The main motivations were (1) that the teleparallel formulation was considered to have advantages with regard to the identification of the energy-momentum of gravitating systems, and (2) teleparallel theory can be regarded as a gauge theory of local translations (see,
e.g.,~\cite{Moller:1961jj,Cho:1975dh,Nester1989,Hehl:1994ue,Itin:1999wi,deAndrade:2000kr,Blagojevic:2000pi,Blagojevic:2002du,Hehl:2012pi}).\footnote{For recent study of teleparallelism as a \emph{higher} gauge theory, see \cite{Baez}.} In fact Einstein's general theory of relativity  (GR) can be rewritten in teleparallel language (see, e.g.,~\cite{Hayashi, Pereira, Kleinert, Sonester}).
The result is a theory equivalent to GR known  variously as the \textit{Teleparallel Equivalent of General Relativity}  (TEGR) or GR$_{\|}$. 
Recently TEGR has been generalized to $f(T)$ theory, a theory of modified gravity formed in the same spirit as generalizing general relativity to $f(R)$ gravity~\cite{0,1}. 

It seems however that observational data puts rather tight constraints on $f(T)$ gravity, effectively reducing the theory (at least for some choices of $f(T)$) back to TEGR (However, see also \cite{CRC} and \cite{ABCL}). This includes constraints from considering a varying fine structure constant~\cite{haowei}, the formation of large-scale structures~\cite{baojiu} and a varying gravitational constant~\cite{haowei2}, as well as solar system constraints ~\cite{lorenzo}. This could hint at problems with $f(T)$ gravity at the theoretical level.

In this paper we will briefly review TEGR and more general theories of teleparallel gravity, including $f(T)$ gravity, in section~\ref{review}. 
In section~\ref{degree}, we discuss  
the degrees of freedom for $f(T)$ theory. We know that there are extra degrees of freedom in the theory beyond those of GR; that the number of extra degrees of freedom is generically 3 was clearly revealed by the Dirac constraint analysis carried out by Li et al.~\cite{Li}.
Despite this, in a careful linear perturbative analysis up to second order on flat Friendmann-Lema\^itre-Robertson-Walker (FLRW) background, no extra mode appears~\cite{IzumiOng}. In the case of Minkowski background, this had already been pointed out in~\cite{Li}.
Disappearance of extra degrees of freedom is known to happen in some gravity theories (e.g., in massive gravity, 
see~\cite{Gumrukcuoglu:2011ew});
this can be caused by the accidental disappearance of kinetic terms at the linear level. If so,  as we shall argue, it is likely to be a signal of superluminal propagation.
We will also discuss the related danger of the \emph{nonlinear constraint effect}, which plagued some other theories of gravity in the past.
Our conclusion is that $f(T)$ theory is very \emph{likely} to give rise to unwanted problems  due to the effects of nonlinear constraints, in addition to propagation outside of
the null cone. As a consequence, the number and type of degrees of freedom could vary as a physical system evolves. This is further explored via the method of characteristics in section~\ref{method}.

The method of characteristics is properly introduced in section~\ref{method}, where we also provide an illuminating example in an application to the nonlinear
Proca field. The same example is given a Hamiltonian analysis that clearly shows the close connection of superluminal propagating modes and nonlinear constraints. Although the characteristic equations allow one to obtain superluminal propagating degrees of freedom, one must be careful to distinguish between \emph{physical} degrees of freedom and \emph{gauge} degrees of freedom, as the example of the Maxwell field illustrates.

In section~\ref{f(T)}, we investigate the dynamical equations in $f(T)$ theory and derive a system of partial differential equations that describe the characteristics of the theory. This calculation reveals nonlinear features that are very prone to allow superluminal propagation outside of the metric null cone.\footnote{We should remark here that \emph{superluminal} propagation may or may not be \emph{achronal}, so it is not certain that such propagation really does lead to a violation of causality. We make further comments on this in section \ref{method}.}
Furthermore, we also show that the constant-time hypersurface in FLRW spacetime, and as a special case Minkowski spacetime, can be a Cauchy \emph{horizon}, 
 which means that even the infinitesimal future time evolution from it is not unique.
We conclude with some discussion in section~\ref{discuss}.

\section{Teleparallel Theories of Gravitation}\label{review}

Let $M$ be a semi-Riemannian manifold of signature $(-,+,+,+)$ with metric $g$.  The teleparallel theories we wish to 
discuss\footnote{There are more general teleparallel theories which also have \emph{nonmetricity}~\cite{Hehl:1994ue}; in fact there are
\emph{symmetric} teleparallel theories---with vanishing torsion but non-vanishing nonmetricity---including one that is
equivalent to GR~\cite{Nester:1998mp,Adak:2005cd}.}
use the vierbein field or tetrad $e_a(~x^\mu)$ as the fundamental dynamical variables. We use Greek letters for spacetime coordinate
indices and Latin letters to label the tetrad vectors. 

Suppose that the manifold is \emph{parallelizable}, i.e., there exist $n$ vector fields $\left\{v_1,...,v_n\right\}$ such that at \emph{any} point $p \in M$ the tangent vectors $v_i|_p$'s provide a basis of the tangent space at $p$. Note that this \emph{particular} set of vector fields must be able to span the tangent spaces at all points of the manifold, i.e., it is a set of everywhere linearly independent \emph{globally non-vanishing vector fields}. Then we can view the mapping between the bases of the tangent space in a coordinate frame $\left\{\partial_\mu\right\}$ to that of a non-coordinate frame $\left\{e_a\right\}$ as an isomorphism $TM \to M \times \Bbb{R}^4$. Orthonormality of the tetrad is thus imposed by the Minkowski metric on the tangent space $T_p M \cong \Bbb{R}^{3,1} = (\Bbb{R}^4, \eta)$, related to the metric on the manifold $M$ by $g_{\mu\nu}=\eta_{_{ab}}e^a_{~\mu}e^b_{~\nu}$.

One then equips the manifold $M$ with the metric-compatible Weitzenb\"ock connection~\cite{Weitzenbock, Weitzenbock2} $\overset{w}\nabla$ defined by
\begin{equation}
\overset{w}\nabla_X Y  := (XY^a)e_a, 
\end{equation}
where $Y=Y^a e_a$.
Equivalently, the connection coefficients are
\begin{equation}
\overset{w}\Gamma~^\lambda_{~\mu \nu} = e_a^{~\lambda}\partial_\nu e^a_{~\mu} = - e^a_{~\mu} \partial_\nu e_a^{~\lambda}.
\end{equation}

In local coordinates, the components of the torsion tensor are
\begin{equation}
\overset{w}T{}^\lambda{}_{\mu\nu} := \overset{w}\Gamma{}^\lambda_{~\nu\mu} - \overset{w}\Gamma{}^\lambda{}_{\mu\nu}=e_a^{~\lambda} (\partial_\mu e^a_{~\nu} - \partial_\nu{}e^a_{~\mu})\neq 0.
\end{equation}
The Weitzenb\"ock connection can be shown to be curvature-free.\footnote{There is an alternative formulation for teleparallel theories that has some advantages; it uses the teleparallel connection as an \emph{a priori} independent dynamical variable that is constrained to have vanishing curvature via a Lagrange multiplier~\cite{Kopczynski,Hehl:1994ue,Blagojevic:2000pi,Blagojevic:2000qs,Blagojevic:2002du,Obukhov:2002tm}.  Here, to more simply communicate our ideas to the audience we most want to reach, we will use the formulation familiar in the $f(T)$ works.}

The Weitzenb\"ock connection differs from the Levi-Civita connection by the \textit{contortion tensor}\footnote{We remark that the proper term is \emph{contortion} instead of the often used ``contorsion''. See e.g., \cite{Hehl}.}
\begin{equation}\label{contortion}
\overset{w}K{}^{\mu\nu}{}_\rho=-\frac{1}{2}\left(\overset{w}T{}^{\mu\nu}{}_\rho - \overset{w}T{}^{\nu\mu}{}_\rho - \overset{w}T_\rho{}^{\mu\nu}\right).
\end{equation}

The Einstein-Hilbert action of general relativity is (in units where $c=1$)
\begin{equation}
S = \frac{1}{2\kappa} \int{d^4x ~\sqrt{-g} R}, \qquad \kappa = 8\pi G.
\end{equation}
Up to a total divergence (which does not affect the field equations) this action is equivalent to
\begin{equation}
S=-\frac{1}{2\kappa} \int{d^4x ~|e| \overset{w}T},\label{TEGRaction}
\end{equation}
where 
$e=\det{(e^a_{~\mu})}$,
\begin{equation}
\overset{w}S_\rho{}^{\mu\nu}:= \frac{1}{2} \left(\overset{w}K{}^{\mu\nu}{}_\rho + \delta^\mu_\rho \overset{w}T{}^{\alpha\nu}{}_{\alpha} - \delta^\nu_\rho \overset{w}T{}^{\alpha \mu}{}_{\alpha} \right),
\end{equation}
and 
$\overset{w}T$ is a certain quadratic-in-torsion scalar:
\begin{equation}\label{T}
\overset{w}T := \overset{w}S_\rho{}^{\mu\nu}\overset{w}T{}^\rho{}_{\mu\nu},
\end{equation}
the latter is the so-called \emph{torsion scalar}.

In this teleparallel language, (\ref{TEGRaction}) is called the \textit{TEGR action}. $f(T)$ gravity simply promotes the ``torsion scalar'' in the TEGR Lagrangian to a function, 
i.e. the action of $f(T)$ gravity is\footnote{For simplicity here we follow the somewhat non-standard convention used in \cite{IzumiOng}, in which the factor $-(2\kappa)^{-1}$ is absorbed into $f$, i.e., TEGR corresponds to $f(T)=-(2\kappa)^{-1}T$.} 
\begin{eqnarray}
S=\int d^4x ~ |e| f(T).
\end{eqnarray}

One can then obtain the field equation for $f(T)$ theory of modified gravity by varying with respect to the frame $e^a{}_{\nu}$, this leads to
 \begin{equation}\label{vfe}
X_a^{~\nu}:=
 \partial_\mu (ef_TS_a^{~\mu\nu}) - ee_a^{~\lambda} T^\rho_{~\mu\lambda} S_\rho^{~\nu\mu}f_T 
- \frac{1}{4}ee_a{}^\nu f = 0.
\end{equation}

We remark that the (quadratic-in-)torsion scalar defined above is a very special one. We note that mathematically we can define a more general torsion ``scalar'' $\tilde{T}$ by relaxing the coefficients:\footnote{In fact one could also include two more quadratic-in-torsion combinations which have odd parity; they can be found in Baekler and Hehl~\cite{Baekler:2011jt}.}
\begin{equation}\label{abc}
\tilde{T} = a\overset{w}T{}^\rho{}_{\eta\mu}\overset{w}T_\rho{}^{\eta\mu} + b\overset{w}T{}^\rho{}_{\mu\eta}\overset{w}T{}^{\eta\mu}{}_\rho + c\overset{w}T_{\rho\mu}{}^\rho \overset{w}T{}^{\nu\mu}{}_\nu.
\end{equation}
In particular one can construct the one-parameter combination  
\begin{equation}\label{OPTP}
\overset{~}T[\lambda] =  \frac{1+2\lambda}{4}\overset{w}T{}^\rho{}_{\eta\mu}\overset{w}T_\rho{}^{\eta\mu} + \frac{1-2\lambda}{2}\overset{w}T{}^\rho{}_{\mu\eta}\overset{w}T{}^{\eta\mu}{}_\rho - \overset{w}T_{\rho\mu}{}^\rho \overset{w}T{}^{\nu\mu}{}_\nu,
\end{equation}
where $\lambda \geq 0$.
Years ago a teleparallel theory using this as the Lagrangian was proposed as an alternative to GR.
This theory passed all observational (as of, perhaps up to, 1990)~\cite{SS, NH} and some theoretical tests~\cite{KN}.
But then certain theoretical problems were found, which we will discuss in the next section.

 We summarize how teleparallel theories work: We start with choosing a good local orthonormal frame $\left\{e_a\right\}$ which is \emph{declared} to be covariantly constant under parallel transport by imposing the Weitzenb\"ock connection. This defines a global frame field that sometimes goes by the name \emph{orthoteleparallel frame} or OT frame. Recall that there are infinitely many possible choices of frames that span the same tangent space at $p \in M$, related to each other by elements of the proper orthochronous Lorentz group. The torsion tensor and consequently the torsion scalar, can thus be expressed in terms of any one of these infinitely many bases. Once we get a parallelization, it is defined up to \emph{global} Lorentz transformation.
This defines a Weitzenb\"ock geometry. The set of all possible parallelizations of $M$ is then partitioned into equivalence classes consisting of parallelizations that relate to each other in the same class by a global Lorentz transformation (thus, one \emph{cannot} pass from one parallelizations to another by a local Lorentz transformation). Note that at this stage, nonzero torsion \emph{does not} mean that we have some ``distinguished frames'' or preferred frames; any equivalence class of frames are on the same footing, in the sense that they parallelized the spacetime equally well. However, once we introduce the physics via the action, we get field equations. \emph{Now} we can only admit certain classes of parallelizations, just as in general relativity where we can only admit metrics that satisfy the Einstein Field Equations.

We remark that due to the lack of local Lorentz invariance (See also the discussions in~\cite{Barrow1} and~\cite{Barrow2}), unlike in general relativity where we can change coordinate systems and use any frame fields, this is not so in generic teleparallel theories. As a consequence, this implies that we cannot directly extract the tetrad from the metric in the straightforward way. TEGR is special since its action does not determine the admissible OT frame but only the metric. From now onwards, we will suppress all the explicit overscript $w$'s on the torsion scalar and connection coefficients etc.

\section{Extra Degrees of Freedom and the Danger of Nonlinear Constraints}
\label{degree}

We are reminded of the problem faced by the one-parameter teleparallel theory in which the teleparallel Lagrangian has the form ~(\ref{OPTP}). This theory has 4 extra degrees of freedom~\cite{Nester3}. Kopczy\'nski~\cite{Kopczynski} first pointed out that this theory has predictability problems; he argued that one is
unable to determine uniquely the evolution of the teleparallel geometry. However, further analysis~\cite{Nester2} showed that the problem is not generic, instead it only occurs for a special class of solutions. Cheng et al.\ verified this behavior using the Hamiltonian approach and identified the troubles 
as a certain effect of nonlinear constraints, there referred to as \emph{constraint bifurcation}~\cite{Nester3}, in which the chain of constraints could ``bifurcate'' depending on the values of the fields. That is to say, the number or type of constraints would depend on the values of the phase space variables. One could imagine that the evolution of some system as depicted by a curve in phase space could pass through regions with a different number or type of constraints. This would be rather strange (since this means that the number of gauges and \emph{physical} degrees of freedom also change), and one will find it hard, if at all possible, to predict the evolution of the system. To put it in another way, the matrix with the Poisson brackets of the constraints as entries may not have constant rank, due to the nonlinearity of the constraints. The number or type of constraints may change as one approaches a point in the phase space where the rank changes. The more general Poincar\'e gauge theory also has similar problems~\cite{Nester8,Nester9,Nester10}. We suspect that this might also be the case for $f(T)$ theory (and probably for many teleparallel theories). It is of course conceivable that with a specific choice of $f(T)$, we could prevent the solutions from ever approaching such problematic regions in phase space, e.g., if they would act as a dynamical \emph{repeller}.

The number of degrees of freedom of a generic $f(T)$ theory of gravity (in $(3+1)$-dimensions) has been found to be five by Miao Li et al.~\cite{Li} using a Hamiltonian approach.\footnote{His analysis was based on the Hamiltonian formulation of Maluf~\cite{Maluf:2000ag,Maluf:2001rg}.  Additional insight could probably be found using the other more general Hamiltonian approaches to gravity theories which specialize to teleparallel theory~\cite{Blagojevic:2000pi,Blagojevic:2000qs,Nester:1991yd,Chen:1998aw}.}  That is, there are three extra degrees of freedom, which the authors suggest could correspond to one massive vector field or one massless vector field with one scalar field.
To understand the properties of the degrees of freedom in a theory, one should first try to
analyze them by a perturbative approach. However, it could happen that not all degrees
of freedom show up in the perturbation level.

For the case of $f(T)$ gravity, as remarked by Li et al.  at the end of their paper~\cite{Li}, for some special backgrounds 
(including Minkowski space with $e^a{}_{\mu}=\delta^a{}_\mu$), 
some of the extra degrees of freedom do not appear in the \emph{linear} perturbation. It is also found that~\cite{IzumiOng}, on the flat FLRW background with a scalar field, linear perturbation up to second order does not reveal \emph{any} extra degree of freedom. In other words, $f(T)$ theory is \emph{highly nonlinear}:
the number and type of constraints in the linear theory is \emph{different} from that of the full, nonlinear theory. This is also one of the problematic features discussed in the context of Poincar\'e Gauge Theory~\cite{Nester8, Nester9, Nester10}. The worry of such behavior is that linearized modes which are ``good'' may cease to be so in the full nonlinear theory, perhaps accompanied by anomalous characteristics, as we will further comment on in the later sections.

The revealing Hamiltonian discussion of Li et al.\ for $f(T)$ gravity in (3+1) dimensional spacetime is rather complicated and some of the results are only inferred implicitly.  However they also worked out \emph{explicitly} the analogous but somewhat simpler Hamiltonian formulation for $f(T)$ gravity in $(2+1)$-dimensions; for this case the problem we wish to address is clearly evident. In this case there are six first class constraints\footnote{We do not explain notations of Li et al.~\cite{Li} here since we will not need them in detail.} $(H, H_i, \Pi^{a0})$, where $a=0,1,2$ and $i=1,2$, as well as four
second class constraints $(\Gamma^1, \Gamma^2, \Gamma^{12}, \pi)$, which leads to two degrees of freedom. 

They define the quantities
\begin{flalign}
&y_i = \left\{H_0, \Gamma^i\right\}, ~y_3 = \left\{H_0, \Gamma^{12}\right\}, \notag \\
&x_0 = \left\{H_0, \pi \right\}, ~x_i = \left\{\Gamma^i, \pi \right\}, ~x_3 = \left\{\Gamma^i, \pi \right\},
\end{flalign}
and
\begin{eqnarray}
A_i &=& \left\{\Gamma^i, \Gamma^{12}\right\}  \nonumber \\
&\approx& 2e\left[g^{0i}\left(g^{01}g^{2m}-g^{02}g^{1m}\right) +g^{1i}(g^{0m}g^{02}-g^{2m}g^{00}) \right. \nonumber \\
&-& \left. g^{i2}(g^{0m}g^{01}-g^{1m}g^{00})\right)]\partial_m\phi.
\end{eqnarray}
The self-consistency equation in matrix form is shown to be $M_{3D}\Lambda_{3D}=0$, where $\Lambda_{3D}=(1,\lambda_1,\lambda_2,\lambda_3,\lambda)^T$ and
\begin{equation}
M_{3D}=
\begin{pmatrix}
0 &  y_1  & y_2 & y_3 & x_0\\
-y_1  &  0 & 0 & A_1 & x_1\\
-y_2 & 0 & 0 & A_2 & x_2\\
-y_3  & -A_1 & -A_2 &  0 & x_3\\
-x_0 & -x_1 & -x_2 & -x_3 & 0\\
\end{pmatrix}
\end{equation}
satisfies $\det{M_{3D}}=0$. This matrix generically is of rank 4. However, as previously mentioned, some of the Poisson brackets could vanish in some cases, which would result in rank changes of the matrix as the system evolves.  Miao Li et al.\ found the explicit formulas for the Lagrange multipliers:
\begin{equation}
\lambda_1 = \frac{A_2x_0 + x_3y_2 - x_2y_3}{A_1x_2 - A_2 x_1},
\end{equation}
\begin{equation}
\lambda_2 = \frac{-A_1x_0 - x_3y_1 + x_1y_3}{A_1x_2 - A_2 x_1},
\end{equation}
\begin{equation}
\lambda_3 = \frac{y_1x_2 - y_2x_1}{A_1x_2 - A_2 x_1},
\end{equation}
\begin{equation}
\lambda = \frac{A_1y_2 -A_2y_1}{A_1x_2 - A_2 x_1}.
\end{equation}
While they were interested in the most general case and thus only considered the case with $A_1x_2 - A_2 x_1\neq 0$, this type of special case is \textit{precisely} what we are interested in in this work. Indeed one sees that in principle $A_1x_2 - A_2 x_1$ can vanish, yet the numerators of these expressions are not generally constrained to vanish at the same time, which allows one or more of the Lagrange multipliers to become unbounded. As we will explain below, this is signaling a superluminal propagation mode.

\section{The Method of Characteristics} 
\label{method}

A viable theory of gravity should satisfy certain theoretical criteria, including the lack of tachyonic modes, or modes that carry negative kinetic energy (i.e., ghosts).\footnote{However, one may still argue that theories with such seemingly pathological features may still be acceptable. We will make further comments in section~\ref{discuss}.} The theory should also support a well-posed initial value problem, that is, satisfy the Cauchy-Kowalevski theorem~\cite{Cauchy, Kowalevski, Nakhushev}. Moreover any propagation mode in the theory should also be described by hyperbolic quasi-linear partial differential equations with well-behaved characteristics, i.e., the characteristic
surfaces should be non-spacelike. One should also consider a theory with a ``good'' Minkowski limit as preferable to a theory that does not have such behavior.

The method for studying characteristics is well-known; see for example the detailed classical work of Courant and Hilbert~\cite{Courant}, or Chapter 8 of the more recent text of~\cite{Strichartz}. A discussion of the method of characteristics in the case of general relativity can be found in Lecture 14 of~\cite{Buchdahl}. For more rigorous treatment see ~\cite{Friedlander}, as well as in \cite{Christodoulou}. 
See also~\cite{Zwanziger} for some useful discussions in the context of external fields in gauge theories. Here we only give a brief summary. Recall that if $P$ is a linear differential operator of order $k$, then we can consider $P$ as a polynomial in the derivative $D$. In multi-index notation we can write this as
\begin{equation}
P = \sum_{|\alpha|\leq k} a_{\alpha} (x) D^\alpha.
\end{equation}
We may ask in what directions it is really of order $k$. For an ODE $a_k(x)(d/dx)^k + ... + a_0(x)=0$ for example, it is obvious that $a_k(x) \neq 0$ is the condition required for the equation to be of order $k$ everywhere. For PDEs with multiple variables, we need to be careful. Consider for example the Laplacian $\partial^2/\partial x^2 + \partial^2/\partial y^2$ in $\Bbb{R}^2$. This is clearly of order 2 in both the $x$- and $y$-direction, however for an operator which is mixed, e.g., $\partial^2/\partial x \partial y$, it is \emph{not} of order 2 in either the $x$- or $y$-direction. Nevertheless, it \emph{is} a second-order operator, which can be revealed by introducing new variables, $t=x+y, s=x-y$, which renders $\partial^2/\partial x\partial y=(1/4)(\partial^2/\partial t^2 - \partial^2/\partial s^2)$. That is, this operator is of order 2 in the $s$ and $t$ directions. In general then, given any operator $P$ at $x$, and a direction $v$, we can make an orthogonal change of variable so that $v$ points along one of the new coordinate axes, say the $x_1$-axis. Now, if the coefficient of the partial derivative $(\partial/\partial x_1)^k$ (in the new coordinate system) is nonzero at $x$, then we can say that $P$ is of order $k$ at $x$ in the direction $v$. We refer to such situation as \emph{noncharacteristic}. That is, \emph{characteristic} refers to the case in which the coefficient of $(\partial/\partial x_1)^k$ vanishes at $x$.

The \emph{total symbol} or \emph{top-order symbol} of $P$ is simply a polynomial obtained by replacing the derivative $D$ with a variable, say $\xi$. The \emph{principal symbol}, denoted $\sigma_P(\xi)$, is the highest degree component of the total symbol. That is,
\begin{equation}
\sigma_P(\xi) = \sum_{|\alpha|=k} a_\alpha \xi^\alpha.
\end{equation}
The principal symbol almost completely determines the qualitative behavior of the solutions of the system. Furthermore, it is well known that for hyperbolic (as well as parabolic) partial differential equations, the zeros of the principal symbol describe the characteristics of the system. That is to say, \emph{the characteristic directions are exactly those for which the principal symbol vanishes}. Indeed the ``standard procedure'' is to begin with the equation of motion or the field equation, keep only the highest derivative terms, and then replace the said derivatives $\partial_\mu \partial_\nu \cdots \partial_\kappa$ with components of the normal vector to the characteristics $k_\mu k_\nu \cdots k_\kappa$ and set the equation to zero. This is called the \emph{characteristic equation}. We can study how the characteristics propagate by looking for what sorts of vectors are allowed as solutions to the characteristic equation.

Note that for a characteristic equation that is a \emph{matrix} equation,
called the \emph{characteristic matrix},
it suffices to consider the vanishing of the determinant of the characteristic equation (known as the \emph{characteristic determinant}) instead of the vanishing of the equation itself (the latter is of course a stronger statement).

In this paper, we are interested in one (or a few) special characteristic direction that gives the signal of a superluminal mode.
By analyzing the characteristic determinant, we can examine all directions at once.

The \emph{characteristic surface} which is orthogonal to the characteristic direction coincides with the edge of Cauchy development, i.e., the Cauchy horizon.
This is because the higher order derivative term with respect to the characteristic direction disappears.
The disappearance of the higher order derivative term causes the evolution to be singular.
Thus, the Cauchy development of the characteristic surface is only the characteristic surface itself and the evolution from it is meaningless even if it is spacelike.

Note that for application in theories of gravity, we would like our characteristic
directions to be null or spacelike, as a timelike characteristic direction is the signal of superluminal propagation
(and could violate causality, although this is not necessarily so even in theories with local Lorentz invariance~\cite{Bruneton, Afshordi, Geroch1}). Correspondingly, there should \emph{not} be any non-trivial solution that corresponds to any \emph{timelike} vector.

\subsection{Example: Nonlinear Proca Field} 
\label{Proca}

The characteristic method is best illustrated via an example or two.

Consider the Lagrangian of the Proca field with a nonlinear term~\cite{Velo:1970ur}:
\begin{equation}
\mathcal{L}=-\frac{1}{4}F^{\mu\nu}F_{\mu\nu} - \frac{1}{2}m^2A^\mu A_\mu - \frac{1}{4}\lambda(A^\mu A_\mu)^2,
\end{equation}
where $F_{\mu\nu} = \partial_\mu A_\nu - \partial_\nu A_\mu$ and $m, \lambda = \text{const.}$ The signature is $(-,+,+,+)$. The field equations are
\begin{equation}
\partial_\mu F^{\mu\nu} - m^2A^\nu - \lambda A^\mu A_\mu A^\nu = 0. \label{nlprocaeq}\\
\end{equation}
 Upon taking the divergence of this equation, because the first term vanishes by the antisymmetry of $F^{\mu\nu}$, we obtain an implicit constraint satisfied by the system:
\begin{equation}
(m^2 + \lambda A^\mu A_\mu)\partial_\nu A^\nu + 2\lambda A^\mu A^\nu \partial_\nu A_\mu = 0.
\end{equation}
For $\lambda = 0$, this gives $\partial_\nu A^\nu = 0$, which upon substituting back into the field equation (\ref{nlprocaeq}) gives
\begin{equation}
(\partial_\mu \partial^\mu - m^2)A^\nu = 0,
\end{equation}
which is the Klein-Gordon equation. The characteristic equation for this case is just $k_\mu k^\mu = 0$, i.e., the characteristic is null.

To find the characteristics in the general $\lambda$ case, replace the highest derivative terms $\partial^NA$ by $k^N\tilde A$.  Here $\tilde A^\mu$ means it is not the value of the vector $A^\mu$, but represents the change of the vector in a certain direction.

We then get the relations
\begin{eqnarray}
k\cdot k\tilde A^\nu-k^\nu k\cdot\tilde A&=&0,\label{procafechar}\cr
(m^2+\lambda A\cdot A)k\cdot \tilde A+2\lambda A\cdot k A\cdot \tilde A&=&0.  
\end{eqnarray}
A linear combination of these gives
\begin{equation}
\left[(m^2+\lambda A\cdot A)k\cdot k +2\lambda (A\cdot k)^2\right]k\cdot \tilde A=0.\label{anomolous}
\end{equation}
Now from eq.~(\ref{procafechar}) one can see that modes with $k\cdot \tilde A=0$ propagate with null characteristics, but from eq.~(\ref{anomolous}) one can see that modes with $k\cdot\tilde A\ne0$  generally have non-null characteristics.  The normal to the characteristic surface could be timelike, indicating a superluminal characteristic.  In that case there is a Lorentz frame in which $k^\mu=(1,0,0,0)$. In such a frame  the condition is
\begin{equation}
m^2+\lambda(-3 A_0^2+A_iA^i)=0.\label{procachar}
\end{equation}

We want to emphasize that this 
tachyonic mode of the nonlinear Proca field can be detected using the Hamiltonian formulation, as we now explain. From the Lagrangian, the canonical momenta are
\begin{equation}
\pi^\mu = \frac{\partial \mathcal{L}}{\partial \dot{A}_\mu}.
\end{equation}
That is, explicitly, $\pi^i = F_0{}^i$ and $\pi^0 = 0$; the latter is a \emph{primary constraint}.

The Hamiltonian density, constructed according to the Dirac-Bergmann constraint procedure \cite{Dirac1, Dirac2, Bergmann, HRT}, is~\cite{Nester8}
\begin{flalign}
&\frac{1}{2}\pi^i\pi_i + \frac{1}{4}F^{ij}F_{ij} - A_0 \partial_c \pi^c + \frac{1}{2} m^2(A^iA_i - A_0^2) \\ \nonumber
 &~+ \frac{1}{4}\lambda (A^iA_i - A_0^2)^2+u\pi^0,
\end{flalign}
where the primary constraint $\pi^0 \approx 0$\footnote{Here ``$\approx$'' denotes Dirac's weak equality, i.e., it only holds on the constraint surface within the phase space.} has been included with an unknown Lagrange multiplier $u$. From the Hamiltonian evolution equation 
\begin{equation}
\dot A_0(x)=\{A_0(x),H\}=u(x),
\end{equation}one finds the meaning of the multiplier, it is the missing ``velocity''.  Preserving the primary constraint leads to the secondary constraint:
\begin{equation}
\chi := \partial_c \pi^c + m^2A_0 + \lambda(A^iA_i - A_0^2)A_0 \approx 0.
\end{equation}
The Poisson bracket of the two constraints is
\begin{equation}
\left\{\pi^0(x), \chi(y)\right\} = \left[m^2 + \lambda(A^iA_i - 3A_0^2)\right]\delta^3(x-y).
\end{equation}
Generically this is non-vanishing, so they make up a 2nd class pair.

However, there is an important exception when the RHS vanishes, \emph{this is exactly the same as the anomalous characteristic condition} (\ref{procachar}).
The dynamical consequence shows up when we require preservation of the $\chi$ constraint, which is the relation that determines the ``unknown multiplier'' $u$:
\begin{equation}
0\approx\dot\chi(x)=\{\chi(x),H\}=u(x)[m^2 + \lambda(A^iA_i - 3A_0^2)] + \mathcal{G},
\end{equation}
where $\mathcal{G}$ denotes a collection of field dependent terms which, generically, are non-vanishing. 
Consequently,  the field ``velocity'' $\dot A_0$, given by the Lagrange multiplier, $u$, becomes unbounded at any point(s) where $\Delta:=m^2+\lambda(A^i A_i-3A_0^2)$ approaches 0.
That is,
\begin{equation}
\lim_{\Delta\to 0}\dot A_0=\lim_{\Delta\to 0} u(x)=\lim_{\Delta \to 0} -\frac{\mathcal{G}}{\Delta} = \infty.
\end{equation}
This is an indication that, with respect to this constant time spacelike hypersurface, there is instantaneous propagation of the $A_0$ mode.
The nonlinear constraint has led to a field-dependent constraint Poisson bracket value, a signal for the occurrence of superluminal propagation, in complete agreement with the previous analysis obtained from the characteristic equation.

\subsection{Example: Scalar Field and Maxwell Field}

One issue that we have to be careful with is the possible presence of gauge degrees of freedom in a theory (which was \emph{not} present in the previous nonlinear Proca field example).

In order to appreciate the issue, first we shall look at the simple example of scalar fields.
We consider the Lagrangian for the single scalar field defined by
\begin{eqnarray}
{\cal{L}}_S = -\frac{1}{2}\partial_\mu \phi_1 \partial^\mu \phi_1. \label{LagS}
\end{eqnarray}
We know that the scalar field $\phi_1$ propagates in the null direction.
In terms of characteristics, we can understand it as follows.
The equation of motion is
\begin{eqnarray}
\partial_\mu \partial^\mu \phi_1=0.
\end{eqnarray}
So the characteristic equation for $\phi_1$ is $k_\mu k^\mu=0$, which means the characteristic is null.

Let us introduce another scalar field $\phi_2$ which does not appear in the Lagrangian.
Trivially, the Lagrangian does not change under the transformation
$\phi_2 \to \phi'_2 = \phi_2 + f(x^\mu)$.
This property is 
 a certain kind of gauge transformation.
The equation of motion which comes from the variation of the Lagrangian (\ref{LagS}) with respect to $\phi_2$ is trivial, i.e., $0=0$.
Thus the characteristic equation for $\phi_1$ and $\phi_2$ becomes $k_\mu k^\mu \times 0=0$,
which is trivially satisfied for any $k_{\mu}$.
The triviality comes from the gauge mode $\phi_2$.
Since $\phi_2$ does not appear in the Lagrangian, we should not include it in the discussion of characteristics.

The situation for any gauge field theory is similar.
For simplicity, we consider the Maxwell theory of the electromagnetic field; the Lagrangian is
\begin{eqnarray}
{\cal{L}}_{\text{EM}} = -\frac{1}{4}F_{\mu\nu}F^{\mu\nu},
\end{eqnarray}
where $F_{\mu\nu}=\partial_\mu A_\nu-\partial_\nu A_\mu$.
The equation of motion is obtained to be
\begin{eqnarray}
\partial_\mu F^{\mu\nu}=0. \label{EOMMax2}
\end{eqnarray}
The principal symbol yields the characteristic equation
\begin{equation}
k_\mu k^\mu \tilde{A}^{\nu} - k_{\mu} k^{\nu} \tilde{A}^{\mu}= 0,
\end{equation}
where we want to seek some nontrivial $\tilde{A}^\lambda$. If $k^\mu$ is null 
 then $k_\mu \tilde{A}^\mu = 0$ gives a propagating mode on the null cone. However, if $k^\mu$ is 
 not null, then the characteristic equation can be satisfied by \emph{any} $\tilde{A}^\lambda=Ck^\lambda$ for arbitrary constant $C$. In particular $k^\mu$ can be chosen to be timelike, which means that the characteristic direction will become timelike.
This superluminal mode however, as we shall see, actually corresponds to gauge mode and so it is not physical.

Alternatively we can consider the determinant method. If we write down the characteristic equation without removing the gauge degree of freedom,
we will obtain characteristic matrix that satisfies
\begin{eqnarray}
\det \left[ k_\mu k^\mu g_{\alpha \beta}- k_\alpha k_\beta \right] =0.
\label{EOMMax}
\end{eqnarray}
The left-hand side of this equation is $0$ for any $k_\mu$, and so this equation becomes trivial.
The origin of the triviality is the gauge degree of freedom.
If we operate $\partial_\nu$ to the left-hand of eq.~(\ref{EOMMax2}), it becomes
algebraically $0$.
Therefore, eq.~(\ref{EOMMax2}) has the information of only three independent equations.
We know the Maxwell field has $U(1)$-gauge, i.e., the Lagrangian is invariant under
the $U(1)$-gauge transformation $A_\mu\to A'_\mu=A_\mu+\partial_\mu \Lambda$.
Since the gauge degree of freedom $\Lambda$ does not appear in the Lagrangian,
the variation of the Lagrangian with respect to the gauge degree of freedom $\Lambda$ must be
trivial, which is the origin of the above trivial equation.
According to the above simple example of scalar fields,
we should eliminate the contribution from gauge degrees of freedom in the discussion of characteristics.
Subtracting the contribution from the gauge degree of freedom, i.e.,
considering the $3\times 3$ matrix whose basis
can be anything independentof $k_\mu$,
we can obtain the exact characteristic equation,
which is satisfied only when $k_{\mu}$ is null. A more rigorous treatment of the removal of gauge degrees of freedom and the characteristics of Maxwell equations can be found in \cite{Christodoulou}.

\section{Characteristics Equations and Problematic Time Evolution in $f(T)$ Gravity}
\label{f(T)}

In this section, we discuss the Cauchy problem for $f(T)$ gravity by an analysis of the characteristics.
First, we shall determine  
 the characteristic equations of $f(T)$ gravity.
After that, we will show that a constant-time hypersurface of the FLRW metric does not provide a good initial condition and that the time evolution of the FLRW metric is not well behaved.
Finally, we concretely show the problematic solution where the time evolution from a FLRW initial condition is not unique.

\subsection{Characteristic Equation of $f(T)$ Gravity} 
\label{characteristic matrix}

From eq.~(\ref{vfe}), we see that the vacuum field equation in $f(T)$ gravity is explicitly given by
\begin{flalign}
0&=\left[\partial_\mu (ee_{a}{}^\rho S_{\rho}^{~\mu\nu}) - ee_a{}^\lambda T^\rho_{~\mu\lambda}S_\rho^{~\nu\mu}\right]f_T - \frac{1}{4}ee_a{}^\nu f\\ \nonumber
 &+ ee_{a}{}^\rho S_{\rho}^{~\mu\nu}(\partial_\mu T)f_{TT}. \label{ftt} 
\end{flalign}
This can be rewritten as 
\begin{widetext}
\begin{equation}
0=f_T(\partial_\mu e)S_a{}^{\mu\nu} + 2e\left[f_T\frac{\partial S_a{}^{\mu\nu}}{\partial T^b{}_{\alpha\beta}}+S_a{}^{\mu\nu}f_{TT}\left(S_b{}^{\alpha\beta}+T^c{}_{\kappa\rho}\frac{\partial S_c{}^{\kappa \rho}}{\partial T^b_{~\alpha\beta}}\right) \right]\partial_\mu \partial_{\alpha} e^b_{~\beta} 
- ee_a{}^\lambda T^\rho_{~\mu\lambda}S_\rho^{~\nu\mu}f_T - e\frac{1}{4}e_a{}^\nu f.
\end{equation}
\end{widetext}

The highest order derivative terms are just
\begin{equation}\label{highest}
2e\left[f_T\frac{\partial S_a{}^{\mu\nu}}{\partial T^b{}_{\alpha\beta}}+S_a{}^{\mu\nu}f_{TT}\left(S_b{}^{\alpha\beta}+T^c{}_{\kappa\rho}\frac{\partial S_c{}^{\kappa \rho}}{\partial T^b_{~\alpha\beta}}\right) \right]\partial_\mu \partial_\alpha e^b_{~\beta}.
\end{equation}
We recall that $S_a^{~\mu\nu}$ is \emph{linear} in $T^b_{~\alpha\beta}$, so we may write
\begin{equation}
S_a^{~\mu\nu} = M_a^{~\mu\nu}{}_b^{~\alpha\beta} T^b_{~\alpha\beta}
\end{equation}
for some $M_a^{~\mu\nu}{}_b^{~\alpha\beta}$ which is \emph{algebraic} in the tetrad $e^c_{~\rho}$.
The explicit form of $M_a^{~\mu\nu}{}_b^{~\alpha\beta}$ is
\begin{eqnarray}
M_a^{~\mu\nu}{}_b^{~\alpha\beta} &=& \frac{1}{4} \eta_{ab} g^{\mu [\alpha} g^{\beta ] \nu}
-\frac{1}{2}e_b^{\ [ \mu} g^{\nu ] [ \alpha}e_a^{\ \beta ]}\nonumber \\ &+& e_a^{\ [ \mu} g^{\nu ] [ \alpha}e_b^{\ \beta ]},
\end{eqnarray}
where the bracket $[\cdots ]$ means antisymmetrization, i.e.,
\begin{eqnarray}
t^{a\cdots [bc]\cdots} := \frac{1}{2} \left( t^{a\cdots bc \cdots} -t^{a\cdots cb \cdots}\right).
\end{eqnarray}
Consequently we have
\begin{equation}
T^c_{~\kappa\rho}\frac{\partial S_c^{~\kappa \rho}}{\partial T^b_{~\alpha\beta}} = T^c_{~\kappa\rho} M_c^{~\kappa\rho}{}_b^{~\alpha\beta} = S_b^{~\alpha\beta}.
\end{equation}
Thus, we can re-write eq.~(\ref{highest}) in a more simplified form as
\begin{equation}\label{char}
2e\left[f_T M_a^{~\mu\nu}{}_b^{~\alpha\beta}+ 2f_{TT}S_a^{~\mu\nu}S_b^{~\alpha\beta}\right]\partial_\mu\partial_\alpha e^b_{\beta}.
\end{equation}
Following the standard procedure we can replace $\partial_\nu\partial_\alpha e^b_{~\beta} \to k_\nu k_\alpha \tilde{e}^b_{~\beta}$
in eq.~(\ref{char}) to obtain the \textit{characteristic equations for $f(T)$ theory of gravity}
\begin{equation}\label{chareqn}
\left[f_T M_a^{~\mu\nu}{}_b^{~\alpha\beta}+ 2f_{TT}S_a^{~\mu\nu}S_b^{~\alpha\beta}\right]k_\mu k_\alpha \tilde{e}^b_{~\beta} = 0; 
\end{equation}
where
\begin{equation}
 M_a^{~\mu\nu}{}_b^{~\alpha\beta}=\dfrac{\partial S_a^{~\mu\nu}}{\partial T^b_{~\alpha\beta}}.\\
\end{equation}
Here the notation $\tilde{e}^b_{~\beta}$ indicates that these are not the values of the frame, but rather represent the change of the frame in a certain direction. 
We note that eq.~(\ref{chareqn}) thus provide $16$ quadratic polynomial conditions on the four $k_\mu$'s.

That is, the characteristic matrix becomes
\begin{eqnarray}
{\cal M}_a^{~\nu}{}_b^{~\beta}\equiv2e\left[f_T M_a^{~\mu\nu}{}_b^{~\alpha\beta}+ 2f_{TT}S_a^{~\mu\nu}S_b^{~\alpha\beta}\right]k_\mu k_\alpha.
\end{eqnarray}
Since the tetrad $e^a_{~\mu}$ always has the inverse $e_a^{~\mu}$,
the property of the matrix ${\cal M}_a^{~\nu}{}_b^{~\beta}$ is the same as that of
\begin{eqnarray}
{\cal M}^{\lambda\nu\gamma\beta}&\equiv&{\cal M}_a^{~\nu}{}_b^{~\beta}e^{a\lambda}e^{b\gamma}
\nonumber\\
&=&2e\left[f_T M^{\lambda\mu\nu\gamma\alpha\beta}+ 2f_{TT}S^{\lambda\mu\nu}S^{\gamma\alpha\beta}\right]k_\mu k_\alpha.\label{charm}
\end{eqnarray}
Therefore for convenience, we shall analyze the matrix in eq.~(\ref{charm}) instead.

Before discussing the property of the characteristic matrix of eq.~(\ref{charm}),
we must first find out the gauge degrees of freedom. We have 16 equations of motion (\ref{vfe}), but not all of them are independent due to gauge degrees of freedom.
We actually have four nontrivial identities.

Indeed, let us consider the invariance of the action under an infinitesimal coordinate transformation:\footnote{Geometrically this represents an infinitesimal diffeomporhism.}
\begin{eqnarray}
x^\mu \to {x'}^\mu = x^\mu + \xi^\mu(x).
\end{eqnarray}
This leads to the transformation of the tetrad $e^a_{~\mu}$ by
\begin{eqnarray}
\delta e^a_{~\mu}=\partial_\mu \left(- e^a_{~\alpha} \xi^\alpha \right)  - T^a_{\ \mu\nu}e_b^{\ \nu}\left(- e^b_{~\alpha} \xi^\alpha \right),
\end{eqnarray}
and we know this is related to a gauge mode, i.e., $-e^a_{~\alpha} \xi^\alpha=F^a$, where $F^a$ are four scalar fields with an internal index.
Varying the action with respect to the component
\begin{eqnarray}
e^{a\ \mbox{\text{(gauge)}}}_{\ \mu}= \partial_\mu F^a - T^a_{\ \mu\nu}e_b^{\ \nu}F^b, 
\label{gauge-e}
\end{eqnarray}
one obtains the nontrivial Noether differential identities \cite{Noether}
\begin{eqnarray}
\partial_\mu \left( X_a^{\ \mu} \right) + e_a^{\ \nu} T^b_{\ \mu\nu} X_b^{\ \mu} =0,
\label{identity}
\end{eqnarray}
where $X_a^{\ \mu}$ is defined in eq.~(\ref{vfe}).
We stress that eq.~(\ref{identity}) is an \emph{identity}, i.e., it is \emph{algebraically satisfied} without using the equation of motion eq.~(\ref{vfe}).

Since the infinitesimal coordinate transformation also leads to a transformation of metric, the
metric has the information about all the gauge degrees of freedom.
Therefore, all the degrees of freedom of the gauge modes can be eliminated by gauge fixing of the metric components.

Now, we move on to the discussion of the characteristic determinant which is constructed by the $12\times12$ part of the characteristic matrix in eq.~(\ref{charm}) having removed the gauge degrees of freedom.
Here, we analyze the characteristic matrix ${\cal M}^{\lambda\nu\gamma\beta}$ by the decomposition into symmetric and antisymmetric parts with respect to $\nu$ and $\lambda$ and also with respect to $\beta$ and $\gamma$. Without loss of generality we can choose an appropriate basis to write down the $12 \times 12$ matrix such that the left half $12 \times 6$ submatrix represents the symmetric part and the right half  $12 \times 6$ submatrix represents the antisymmetric part.

First, we consider the first term of the characteristic matrix in eq.~(\ref{charm}). The first term of the symmetrized characteristic matrix is well behaved because it is the only contribution in TEGR.
The symmetric part of the characteristic equation is related to the metric component, which \emph{a priori} has 10 degrees of freedom. However we know that 4 of the degrees of freedom are gauge modes and we have only 6 physical modes.
We claim that its rank is generically 6, and indeed after removing the gauge degrees of freedom the second order derivatives are associated with all components of the symmetrized tetrad.
We consider next the antisymmetrized component
with respect to $\nu$ and $\lambda$. The first term then gives zero contribution, which means that the rank of the first term is indeed $6$ for general $k_\mu$.

To summarize, the characteristic matrix of $f(T)$ gravity, after removing the gauge degrees of freedom, can be written as a $12 \times 12$ matrix, with four $6 \times 6$ block submatrices.
The generic rank of the upper left block is 6, corresponding to the symmetrized tetrads.

Next, we consider the antisymmetrized components of the second term in eq.~(\ref{charm}).
The form of the second term is the product of two $S^{\lambda\nu\mu}k_\mu$ and its rank is one (since one of the eigenvalue vanishes).
One can thus form a $2 \times 2$ matrix at the very center of the $12 \times 12$ matrix, in which each element belongs to one of the four blocks, respectively. 
Together with the $6 \times 6$ block corresponding to the symmetrized tetrad, we thus get a $7 \times 7$ matrix.

Therefore we have found that the ranks of the first term and the second term of ${\cal M}^{\lambda\nu\gamma\beta}$ are 6 and 1, respectively. By the fact that $S^{[\lambda\nu]\mu} k_\mu$ is generically nonzero, we can prove that the rank of matrix ${\cal M}^{\lambda\nu\gamma\beta}$ is generically 7. This rests on the fact that the determinant of the $7 \times 7$ matrix is generically nonzero, which follows from elementary linear algebra (See Appendix.\ref{Appendix}).

It is important to note that the $2 \times 2$ matrix at the center, being of rank 1, corresponds to only \emph{one} of the extra degrees of freedom in $f(T)$ gravity.  However we know that the theory contains 3 extra degrees of freedom -- so where are the remaining two?

In fact the five antisymmetrized components corresponding to  the tetrad $e^{(i)a}{}_\nu$ ($i=1,\cdots,5$), which satisfy
\begin{eqnarray}\label{5part}
S^{[\lambda\nu]\mu} e_{a\lambda} \partial_\mu e^{(i)a}{}_\nu =0,
\end{eqnarray}
become zero.
This means that in the equations for the five tetrads $e^{(i)a}{}_\nu$ the second order derivative does not appear and that the characteristic equation should be at most a first order differential equation. In other words, the remaining two simply \emph{do not} show up in the characteristic matrix!

Thus, in order to obtain the complete characteristic we must derive the characteristic equation for $e^{(i)a}{}_\nu$ separately. This is what we will do next for the sake of completeness, although it is not a crucial part for our subsequent analysis (since it is sufficient to show that the one degree of freedom that \emph{does} show up in the characteristic matrix, leads to problematic behavior).

In order to derive the characteristic equation for $e^{(i)a}{}_\nu$, we must go back to the equation of motion eq.~(\ref{vfe}),
which can be written as
\begin{eqnarray}\label{eqanti}
\partial_\mu \left( e S_a^{~\mu\nu}f_T\right)-\frac{1}{4}e e_a^{~\nu}f-e e_a^{~\lambda}T^\rho_{~\mu\lambda}S_\rho^{~\nu\mu} f_T=0.
\end{eqnarray}
If we consider equation of motion from the variation of the action with respect to $e^{(i)a}{}_\nu$, we can easily show that the first term in eq.~(\ref{eqanti}) \emph{does not appear}.
The equation of motion (\ref{eqanti}) is obtained from variation of the action
\begin{eqnarray}
\delta S &=& \int d^4x ~\delta e^a_{~\nu} \Bigl[ \partial_\mu  \left(e S_a^{~\mu\nu}f_T\right)
\nonumber\\
&&\qquad\qquad\qquad
-\frac{1}{4}e e_a^{~\nu}f-e e_a^{~\lambda}T^\rho_{~\mu\lambda}S_\rho^{~\nu\mu} f_T
\Bigr]\nonumber\\
&=&- \int d^4x \Bigl[ \left( \partial_\mu \delta e^a_{~\nu} \right)e S_a^{~\mu\nu}f_T 
\nonumber\\
&&\qquad
+\delta e^a_{~\nu}\left(\frac{1}{4}e e_a^{~\nu}f+e e_a^{~\lambda}T^\rho_{~\mu\lambda}S_\rho^{~\nu\mu} f_T\right)\Bigr],
\end{eqnarray}
and for $\delta e^a_{~\nu}= \delta e^{(i)a}{}_\nu$ the first term of the right-hand side becomes zero due to eq.~(\ref{5part}).
The second term in eq.~(\ref{eqanti}) is also zero because of antisymmetrization.
Eventually, only the last term in eq.~(\ref{eqanti}) gives a non-zero contribution.

Although the last term in eq.~(\ref{eqanti}) seems to be a higher order equation of the first order derivative of $e^{(i)a}{}_\nu$ (which, if true, will render the characteristic method unworkable), we can show that it is actually linear. Indeed, we can multiply it by $-e^{a\alpha}/e$ for convenience. We then have the antisymmetrized equation with respect to $\alpha$ and $\nu$
\begin{eqnarray}
\contraction{T^{\rho}}{{}^\alpha}{{}_{\mu}S_\rho}{{}^\nu}
T^{\rho}{}^\alpha{}_{\mu}S_\rho{}^\nu{}^\mu f_T=0,
\end{eqnarray}
where
\begin{eqnarray}
\contraction{T^{\alpha \cdots}}{{}^{\beta}}{{}^{\cdots}}{{}^{\gamma}}
T^{\alpha \cdots}{}^{\beta}{}^{\cdots}{}^{\gamma}{}^{\cdots} :=
\frac{1}{2}\left( T^{\alpha\cdots\beta\cdots\gamma\cdots}- T^{\alpha\cdots\gamma\cdots\beta\cdots}\right).
\end{eqnarray}
This gives the two branches of the equation, $f_T=0$ and
\begin{eqnarray}
\contraction{T^{\rho}}{{}^\alpha}{{}_{\mu}S_\rho}{{}^\nu} T^{\rho}{}^\alpha{}_{\mu}S_\rho{}^\nu{}^\mu =0. \label{antie}
\end{eqnarray}
If we choose the first branch, i.e., $f_T=0$, the five equations become degenerate, and thus, time evolution cannot be fixed uniquely.
This causes trouble regarding the  well-posedness of the Cauchy problem.
Therefore, we choose the second branch.\footnote{
There is however \emph{a priori} neither a physical principle nor a mathematical requirement that requires us to choose the second branch. This might be another issue of $f(T)$ gravity.}
The second order terms of the first derivative in eq.~(\ref{antie}) can be written as
\begin{eqnarray}
g^{\mu\beta}\left(-\frac{1}{2}g^{\gamma[\alpha}e_a^{~\nu]}e_b^{~\rho}+g^{\gamma[\alpha}e_b^{~\nu]}e_a^{~\rho}\right)\partial_{[\beta}e^a_{\rho]}\partial_{[\mu}e^b_{~\gamma]}=0. \label{2nd}
\end{eqnarray}
If both derivative factors are variations of $e^{(i)a}{}_\nu$, eq.~(\ref{2nd}) becomes trivial. Thus at least one of the derivative term is \emph{not} such a variation.
This means that eq.~(\ref{antie}) is not second order but first order with respect to $\partial_\mu e^{(i)a}{}_\nu$ and the characteristic equation for $e^{(i)a}{}_\nu$ can be shown to be
\begin{flalign}
&\contraction{\biggl( \frac{1}{4}g}{{}^{\beta}}{{}^{[\alpha}T^{\nu]\mu}}{{}^{\lambda}}
\contraction{ \biggl( \frac{1}{4}g{}^{\beta}{}^{[\alpha}T^{\nu]\mu}{}^{\lambda}
+g}{{}^{\alpha}}{{}^{[\beta}T^{\lambda]\mu}}{{}^{\nu}}
\biggl( \frac{1}{4}g{}^{\beta}{}^{[\alpha}T^{\nu]\mu}{}^{\lambda}
+g{}^{\alpha}{}^{[\beta}T^{\lambda]\mu}{}^{\nu} 
+ \frac{1}{2}T^{[\alpha\nu][\beta}g^{\lambda]\mu}
+\frac{1}{2}T^\mu g^{\alpha[\lambda}g^{\beta]\nu} \nonumber \\ &+\frac{1}{2}T^{[\beta}g^{\lambda][\nu}g^{\alpha]\mu}\biggr) e_{a\beta} k_\mu \dot e^{(i)a}{}_\lambda =0, \nonumber\\
\label{charaanti}
\end{flalign}
where the notation $\dot e^{(i)a}{}_\lambda$ is similar to $\tilde e^a_{~\nu}$, i.e.,
it also represents the change of the frame in a certain direction.
Eq.~(\ref{charaanti}) gives the characteristic matrix for $e^{(i)a}{}_\nu$:
\begin{eqnarray}
{\cal M}_{(\text{anti})}^{\alpha\nu\beta\lambda} &=&
\contraction{\biggl( \frac{1}{4}g}{{}^{\beta}}{{}^{[\alpha}T^{\nu]\mu}}{{}^{\lambda}}
\contraction{ \biggl( \frac{1}{4}g{}^{\beta}{}^{[\alpha}T^{\nu]\mu}{}^{\lambda}
+g}{{}^{\alpha}}{{}^{[\beta}T^{\lambda]\mu}}{{}^{\nu}}
\biggl( \frac{1}{4}g{}^{\beta}{}^{[\alpha}T^{\nu]\mu}{}^{\lambda}
+g{}^{\alpha}{}^{[\beta}T^{\lambda]\mu}{}^{\nu} \nonumber \\
&+&\frac{1}{2}T^{[\alpha\nu][\beta}g^{\lambda]\mu}
+\frac{1}{2}T^\mu g^{\alpha[\lambda}g^{\beta]\nu}+\frac{1}{2}T^{[\beta}g^{\lambda][\nu}g^{\alpha]\mu}\biggr) k_\mu. \nonumber\\
{}
\end{eqnarray}
Then, we have the exact characteristic matrix which is a combination of the $7\times 7$ matrix 
${\cal M}^{\alpha\nu\beta\lambda}$ and the $5\times 5$ matrix ${\cal M}_{(\text{anti})}^{\alpha\nu\beta\lambda}$.%
\footnote{
The cross terms of the 5 components $e^{(i)a}{}_\nu$ and the other 7 do not contribute to the characteristic matrix, since their components have only one derivative.
Since the nonzero components of the $7\times 7$ matrix ${\cal M}^{\alpha\nu\beta\lambda}$ have two derivatives, the contribution of the cross terms with  any smaller number of derivatives should not appear in the characteristic equation.}

To summarize, the remaining two extra degrees of freedom do not show up in the characteristic matrix, since their corresponding characteristic equation only contains \emph{one} derivative term.  They can nevertheless be described by another $5 \times 5$ characteristic matrix ${\cal M}_{(\text{anti})}^{\alpha\nu\beta\lambda}$. We remark that despite the single derivative term here, this \emph{does not} necessary mean that the \emph{dynamics} of the two extra degrees of freedom is governed by first order equations. This may be analogous to the Dirac equation, which is first order and linear. However, together with its conjugate equation, one can derive a \emph{second order} Klein-Gordon equation governing the dynamics of the system. Likewise, Maxwell's equations are first-order but combining them gives the second order wave equation.

In the next section we will give concrete examples to demonstrate problematic behaviors that could arise in $f(T)$ gravity due to bad characteristics.
To appreciate what is likely to go wrong with the characteristic, let us look at the characteristic equation eq.~(\ref{chareqn}) again.
We see that the first term is good since it is the only term present in TEGR, and we know that characteristics are null in general relativity. In fact the first term gives dynamics to the 6 metrical components of $\tilde{e}^b_{~\beta}$. The second term quadratic in the field strength
$S^a_{~\mu\nu}k^\mu$ governing \emph{one} of
the extra degrees of freedom however is almost certainly going to give a disaster if $S^a_{~\mu\nu}k^\mu=0$ for timelike $k^\mu$.
It suggests that with some field values the characteristics could leak outside the metric null cone, which in turn means that there might be problems with tachyonic modes. This nonlinear behavior is reminiscent of the nonlinearity we encountered in the generalized Proca field example.


Note that the case for TEGR is very different since the theory is locally Lorentz invariant, which allows us to choose a gauge that simplifies a lot of the calculation. In any case, TEGR corresponds to $f_{TT}=0$
for any value of $T$ and there is no problem with the characteristics.

\subsection{Absence of Cauchy Development from a Constant-Time Hypersurface in FLRW Geometry}

In this section, we will show that in $f(T)$ gravity the time direction, i.e., $(\partial/\partial t)^\mu$, in the FLRW metric is the characteristic direction, and thus the time evolution of the FLRW metric is not unique \emph{even in the infinitesimal future}.
We analyze the characteristic equation in $f(T)$ gravity which has been derived in the previous subsection.
The characteristic determinant can be decomposed into two parts, the $7\times 7$ matrix ${\cal M}^{\alpha\nu\beta\lambda}$ and the $5\times 5$ matrix ${\cal M}_{(\text{anti})}^{\alpha\nu\beta\lambda}$.
In our analysis we look into the $7\times 7$ matrix ${\cal M}^{\alpha\nu\beta\lambda}$, which has the information about all of the six symmetric components and one of the antisymmetric components.
The antisymmetric component appears only in the second term of the characteristic matrix in eq.~(\ref{charm}), and thus for some $k^\mu$, the determinant becomes zero if the value of its antisymmetric component is zero.
This means that $k^\mu$ is one of the characteristic directions.
To see this, it is enough to solve $S_{[\lambda\nu]\mu}k^\mu=0$ for $k^\mu$.
In the case of the FLRW metric, $k^\mu=(\partial/\partial t)^\mu$ is a solution.

Let us see this explicitly.
We consider the cotetrads
\begin{eqnarray}
&&e^0_{\ \mu}dx^\mu= dt,\nonumber\\
&&e^a_{\ \mu}dx^\mu= a(t) \delta^a_i dx^i. \label{FLRWtetrad}
\end{eqnarray}
These cotetrads give the FLRW geometry with flat spatial section:
\begin{eqnarray}
&&ds^2 = g_{\mu\nu}dx^\mu dx^\nu= -dt^2 + a^2(t) \delta_{ij}dx^i dx^j. \label{FLRW}
\end{eqnarray}
Given the cotetrads, we can calculate the torsion tensor $T^\lambda_{~\mu\nu}$ and $S^{\lambda\mu\nu}$.
They are 
\begin{eqnarray}
&&T^i_{\ tj} = H  \delta^i_j,\\
&&S^{ijt} = - a^{-2}H \delta^{ij},
\end{eqnarray}
while the other components all vanish.
Here $H$ is the Hubble parameter, i.e., $H= (d a/dt) /a$.
Then, it is trivial that $k^\mu=(\partial/\partial t)^\mu$ is a solution of $S^{[\lambda\nu]\mu}k_\mu=0$.

\subsection{Concrete Example of Non-Uniqueness of Time Evolution in FLRW Spacetime}

We have shown that the FLRW solution has a spacelike characteristic hypersurface, and thus we cannot properly discuss the time evolution from a constant time hypersurface even into the infinitesimal future.
Here, we show a concrete solution where the time evolution from the FLRW initial condition without torsion is \emph{not unique} in $f(T)$ gravity with or without matter.

We introduce a matter action  
\begin{eqnarray}
S_m=\int d^4x~ \sqrt{-g} \mathcal{L}_m.
\end{eqnarray}
Then, equation of motion (\ref{vfe}) is modified as
\begin{eqnarray}\label{mattereom}
X_a^{~\nu}=\frac{1}{4} \frac{\delta S_m}{\delta e^a_{~\nu}}.
\end{eqnarray}
Assume that the matter action is locally Lorentz invariant, i.e., it depends on the tetrad only through the metric. Then 
Eq.~(\ref{mattereom}) can be rewritten as 
\begin{eqnarray}
X_\mu^{~\nu}= \frac{1}{2} g_{\mu\alpha} \frac{\delta S_m}{\delta g_{\alpha \nu}}.
\label{vfe2}
\end{eqnarray}

Let us now consider the ansatz
\begin{eqnarray}
&&e^0_{~\mu} dx^{\mu}= \cosh \theta (t) dt + a(t) \sinh \theta (t) dx,\nonumber\\
&&e^1_{~\mu} dx^{\mu}= \sinh \theta (t) dt + a(t) \cosh \theta (t) dx, \nonumber\\
&&e^2_{~\mu} dx^{\mu}= a(t) dy, \label{ansatz}\\
&&e^3_{~\mu} dx^{\mu}= a(t) dz, \nonumber
\end{eqnarray}
which also gives the FLRW metric (\ref{FLRW}). 
For $\theta=0$, this ansatz is equal to eq.~(\ref{FLRWtetrad}).
We note that this is a good choice of tetrad despite the lack of local Lorentz invariance in the theory. In particular the vector fields are everywhere smooth and linearly independent and thus provide a parallelization of FLRW spacetime. This choice of tetrad gives a parallelization that belongs to a different equivalence class of parallelizations than the usual diagonal tetrad of FLRW spacetime (See the penultimate paragraph of section~\ref{review}), but it is nevertheless a solution to the field equation, as we will show below. Note that of course with the choice $a(t)\equiv 1$, our analysis below is in particular applicable to Minkowski spacetime.

The torsion tensor $T^\rho_{~\mu\nu}$ and $S_\rho^{\ \mu\nu}$ are given by
\begin{eqnarray}
&&T^t_{\ tx}
=-T^t_{\ xt}=a(t)\partial_t\theta(t),
\quad
T^x_{\ tx}=-T^x_{\ xt}=H, \nonumber
\\ 
&&
T^p_{\ tq}=-T^p_{\ qt}=H\delta^p_q,
\quad
S_p^{\ x q}=-S_p^{\ qx}= \frac{1}{2 a(t)} \partial_t \theta(t) \delta_p^q,  \nonumber
\\
&&
S_x^{\ t x}=-S_x^{\ xt}=H,
\quad
S_p^{\ t q}=-S_p^{\ qt }=H\delta_p^q,
\end{eqnarray}
where the other components are zero. Here $p,q \in \left\{y,z\right\}$.
From this one readily finds that\footnote{We remind the readers that we work in $(-,+,+,+)$ convention. The torsion scalar is $T=-6H^2$ if one works in the often used convention $(+,-,-,-)$ in the literature of $f(T)$ cosmology.}
\begin{eqnarray}
T=6H^2.
\end{eqnarray}
Substituting the above frame $e^a_{~\mu}$, metric, torsion tensor, torsion scalar and
$S_\rho^{\ \mu\nu}$ into the equation of motion (\ref{vfe2}), 
we can find that $\theta (t)$ does not appear in eq.~(\ref{vfe2}).
This means that if ansatz (\ref{FLRWtetrad}) (i.e. $\theta=0$) originally satisfies eq.~(\ref{vfe2}),
for \emph{any} form of $\theta (t)$ the ansatz (\ref{ansatz}) is indeed an \emph{exact} solution of $f(T)$ gravity.

Now, we can choose the form of the function $\theta (t)$ such that $\theta (t)=0$ for $t<0$ and
$\theta (t) \neq 0$ for $t>0$.
In this solution, on the spacelike hypersurface $t=0$ torsion suddenly emerges seemingly from nothingness,
\footnote{
This is probably due to superluminal propagating modes that come in from infinity. One can compare this to a similar phenomenon in anti-de Sitter (AdS) spacetime. AdS spacetime is not globally hyperbolic since although its (conformal) boundary is infinitely far away as measured along any spacelike path, massless particles can propagate to the boundary and back in finite time. In particular this means that, outside the Cauchy development,
the initial data does not uniquely determine the time evolution, due to information that can flow in from infinity. In our case however, the Cauchy problem is ill-defined not because of a peculiar geometry of the spacetime, but because of physical superluminal propagating degrees of freedom. In fact it is much worse than the AdS case since one can still at least have Cauchy development locally in AdS.}
and the value of $\theta (t) \neq 0$ for $t>0$ is arbitrary.
This means that---even in the infinitesimal future---time evolution is not uniquely determined by the initial data, in other words the Cauchy problem is not well-defined.\footnote{This example of an unpredictable solution has some similarities with an example that demonstrated unpredictability in another teleparallel theory, which appeared in the seminal work of Kopczy\'nski~\cite{Kopczynski}.}\footnote{Let us also remark that one could consider an initial tetrad which is like the one in this example only in a \emph{small region} on the initial spacelike hypersurface.  Such a tetrad makes an even nastier example.} 

Note that this problem does not arise in the case of TEGR due to local Lorentz invariance of the theory.

\section{Discussion}\label{discuss}

To conclude, in this work we have derived the set of partial differential equations that govern the characteristics of $f(T)$ theory of modified gravity, from which we see that there  will generally be a real danger of superluminal propagation that leaks outside of the metric null cone unless $f_{TT}=0$  for any value of $T$,
which is the case of TEGR.

We also took a closer look at the degrees of freedom in the theory. We find that $f(T)$ gravity is a highly nonlinear theory in which many things could go wrong. First of all, as commented by Li et al.\ the linearized theory has a \emph{different} number of degrees of freedom than the full theory, which has three extra degrees of freedom compared to TEGR (and thus GR). Such nonlinear behavior was further confirmed by careful linear perturbation on a flat FLRW background~\cite{IzumiOng}, in which \emph{none} of the three extra degrees of freedom appear even in the second order linear perturbation. In this work we note that the characteristic equation of $f(T)$ gravity, i.e., eq.~(\ref{chareqn}) contains a highly nonlinear second term.

Indeed, we see that generically it excites the six pieces of $\tilde{e}^b_{~\beta}$, essentially $e^{[\alpha\beta]}$.%
\footnote{
In more general theories with tetrads being the fundamental variables, not only $e^{[\alpha\beta]}$ but also $e^0_{~\beta}$ can be dynamical.
}
Of course, a more careful analysis by Li et al.\ showed that there are only 3 extra modes instead of 6. Regardless, we see that
the antisymmetric modes \emph{are} excited
in a highly nonlinear fashion that is very prone to causing troubles. In fact, these extra modes are certain to have a very singular weak field limit, since the linearized theory has only the metric degrees of freedom. The number of degrees of freedom and the number of constraints thus seem to be likely to depend on the amplitude of the field.  As the tetrad fields approach some specific values, some of the Poisson brackets are approaching zero, which in turn means that their corresponding ``velocity'' Lagrange multipliers tend to become unbounded, which signals the occurrence of \emph{instantaneous}, i.e., faster-than-light propagation, much like the situation we have exhibited for the nonlinear Proca field (See also the similar discussion in~\cite{Nester8}.). Moreover these extra dynamic parts are likely to propagate outside the metric null cone and are also likely to allow for negative energy propagating waves. By further analyzing the characteristic matrix of $f(T)$ gravity, we have shown that there is indeed a \emph{physical} superluminal mode that could arise from one of the three extra degrees of freedom. Closely related to this issue, we have also demonstrated that the Cauchy problem is ill-posed in flat FLRW spacetime and even in Minkowski spacetime. We feel that this is a very bad property for any theory of gravity.

This echoes the remark (Lesson 4) in~\cite{f(R)} that extra degrees of freedom in modified gravities are prone to give rise to complications that are very hard to control. Although that remark was made under the context of $f(R)$ gravity, it apparently also holds true for $f(T)$ gravity. We recall that $f(T)$ gravity was often thought of as being more well-behaved and easier to deal with compared to $f(R)$ gravity, since its equation of motion is of second order instead of fourth order. However it now seems that this advantage comes with many trade-offs. For example, unlike its $f(R)$ counterpart, $f(T)$ gravity is not locally Lorentz invariant, and thus it is much more difficult to impose an ansatz for a solution \emph{a priori} except in the simplest situations, such as for a flat spacetime. Furthermore, as we pointed out in this work, many problems that do not arise in $f(R)$ theory could arise in $f(T)$ gravity due to extra degrees of freedom and nonlinearity of the constraints.

We recall that in~\cite{IzumiOng}, it is pointed out that the extra degrees of freedom do not appear at the linear perturbation level on a FLRW background, and that this behavior is similar to that of nonlinear massive gravity~\cite{deRham:2010kj}
in which, while nonlinear analysis shows that there are generically five gravitational degrees of freedom~\cite{Hassan:2011hr},
in the second order action on open FLRW background there are only two tensor degrees of freedom~\cite{Gumrukcuoglu:2011ew, Shinji}.
Since the hidden degrees of freedom in fact cause nonlinear instability in the case of nonlinear massive gravity~\cite{DeFelice:2012mx}, it is a concern that $f(T)$ theory could exhibit a similar pathology. Here we would like to point out another similarity between nonlinear massive gravity and $f(T)$ gravity: they both probably exhibit superluminal propagation modes. The case for nonlinear massive gravity was studied in~\cite{Chien-I}, in which the authors investigated the effect of the helicity-0 mode of the theory and found that energy can probably be emitted superluminously on the self-accelerating background. In the case of the ghost-free Wess-Zumino massive gravity~\cite{WZ}, Deser and Waldron recently showed that it also exhibits superluminal propagation modes~\cite{DW}.

Here we should also make some comments on the possibility of \emph{negative energy solutions} in $f(T)$ gravity. In~\cite{Nester1989} it was shown, directly in terms of the teleparallel variables, that TEGR only permits positive energy. In general relativity, this is a well-known theorem proved initially by Schoen and Yau~\cite{Yau1, Yau2} and subsequently by Witten~\cite{Witten}. In the TEGR formulation the argument was that one can arrange for the Hamiltonian to be dominated by positive terms. A key step in the argument was using the local frame gauge
freedom to remove or at least control the non-positive terms in the TEGR Hamiltonian.  But in $f(T)$ gravity (and teleparallel theory in general) one does \emph{not} have local frame gauge freedom.  Therefore there is no way to kill or control some of the non-positive terms.

In the positive energy test developed in~\cite{Nester1987}, it is argued that while it is hard to prove that the energy for some theory is positive for every solution to the initial value
constraints,  it is not so hard to show that a negative energy solution exists. Indeed it is sufficient to show that the initial value constraints admit a non-trivial zero energy solution.  If so the theory should be discarded.  In principle this is a very strong test of a theory (i.e., it could exclude many theories which pass other tests).  Unfortunately in practice it was not so easy to find \emph{even one} bad solution explicitly, so it is generally hard work getting results beyond those of the linearized theory.  In any case, to us it seems likely that $f(T)$ gravity would be very vulnerable to having a nontrivial non-positive energy solution for most choices of $f$, although it is not expected to be very easy to show this.

Although our work is purely classical, we would also like to make a brief comment regarding strong coupling problem once we consider \emph{quantizing} the theory: Naively, the lack of extra degrees of freedom at the linear level is caused by accidental disappearance of certain kinetic terms at the linear level. 
Vanishing of kinetic terms at the linear level corresponds to the small limit of their coefficients, which in turn means the nonlinear coupling terms of the canonically normalized modes are becoming very strong.
Consequently the perturbative approach cannot be applied and the theory is out of control at the quantum level.
Then there is no reason to trust the classical effective action any more~\cite{ArkaniHamed:2002sp}.
This is expected to be a problem if one attempts to naively quantize $f(T)$ gravity. It is also worth mentioning that existence of ghosts in a classical theory of gravity may be acceptable if the ghost can be pushed to the Planck scale, to be dealt with by quantum theory of gravity \cite{Pisin}.

Finally, it is worth emphasizing that due to the fact that $f(T)$ theory is highly nonlinear, a full rigorous analysis is at best a very difficult task. One could have a viable theory if one manages to find a particular choice of $f$ that could avoid the problems raised in this work. In other words, these seemingly serious ``problems'' could in fact be a blessing in disguise, since they may provide a guide for narrowing down the viable forms of $f$. 

We conclude by the following remark: We agree that non-linear theories merit investigation, but, as far as we can see,  $f(T)$ gravity does not seem to have the right kind of non-linearity; its non-linearity may be too simple.
As commented in \cite{Nester8}, it is possible that  this could be an indication that an \emph{even more nonlinear theory} is required. Recall that the linearized spin-2 theory of gravity has problems that are only cured by nonlinearities of the full theory of general relativity (see, e.g., Route 5 in Box 17.2 in MTW \cite{MTW}, or Feynman's Lectures on Gravitation \cite{Fe}). Furthermore, in Fierz-Pauli massive gravity \cite{FP}, the Vaishtein mechanism \cite{Vainshtein} (introduced to avoid the vDVZ \cite{vDV, Z} discontinuity) excites the Boulware-Deser ghost \cite{BD} due to non-linearity. However, by introducing additional non-linearity, this ghost mode can be killed (this is the so-called "nonlinear massive gravity) \cite{dRGT1, dRGT2, dRGT3, dRGT4}.  Not surprisingly, this newly introduced non-linearity seems to give rise to other problems \cite{DW}. While inherently non-linear theories certainly merit serious consideration; finding one that is free of problems, in particular of the kind we have discussed, seems not so easy.

\acknowledgments
Yen Chin Ong would like to thank Brett McInnes and Wu-Hsing Huang for fruitful discussions about various aspects of differential geometry. Yen Chin Ong and James Nester would also like to thank Friedrich Hehl for providing comments on the early version of the draft. Keisuke Izumi is supported by Taiwan National Science Council (TNSC) under Project No.\ NSC101-2811-M-002-103. Pisin Chen is supported by TNSC under Project No.\ NSC 97-2112-M-002-026-MY3, by Taiwan's National Center for Theoretical Sciences (NCTS), and by the US Department of Energy under Contract No.\ DE-AC03-76SF00515. James M. Nester is also supported by TNSC under Project No.\ 100-2119-M-008-018 and 101-2112-M-008-006
and in part by NCTS. Yen Chin Ong is supported by the Taiwan Scholarship from Taiwan's Ministry of Education.

\appendix
\section{Appendix}
\label{Appendix}

Let us write $M=(M)_{ij}$ to denote a matrix $M$, while $M_{ij}$ denotes the matrix element at the $i^\text{th}$ row and $j^\text{th}$ column of $M$.
\begin{Lemma} Let $M=(M)_{ij}$ be an $n \times n$ matrix with non-vanishing determinant. Let $A=(A)_{ij}$ be a $2 \times 2$ matrix with vanishing determinant. Let $N=(N)_{ij}$ be the $(n+1) \times (n+1)$ matrix constructed from $M$ and $A$, such that
\begin{itemize}
\item[(1)] $N_{nn} = M_{nn} + a_{11}$,
\item[(2)] $N_{n(n+1)}=a_{12}$,
\item[(3)] $N_{(n+1)n} = a_{21}$,
\item[(4)] $N_{(n+1)(n+1)} = a_{22}$,
\item[(5)] $N_{(n+1)i}=0, ~ i \in \left\{1,2,\cdots,n\right\}$,
\item[(6)] $N_{i(n+1)}=0, ~ i \in \left\{1,2,\cdots,n\right\}$,
\item[(7)] $N_{ij}=M_{ij}, ~\text{otherwise}$.
\end{itemize}
Then, $\det{N}=a_{22}\det{M}$.
\end{Lemma}

\begin{proof}
The determinant of $M$ can be calculated straightforwardly by Laplace expansion:
\begin{flalign}\label{1}
\det{N} &= a_{22} \det\left[ \begin{pmatrix} 0 & \cdots & 0 \\ \vdots & \ddots & \vdots \\ 0 & \cdots & a_{11} \end{pmatrix} + (M)_{nn} \right] \nonumber \\ 
&- a_{21}\det[(M)_{n(n-1)}|r]_{n \times n},
\end{flalign}
where $r=(0 ~\cdots 0 ~a_{12})^T$.

The first term is
\begin{equation}
a_{22} \left[ (a_{11} + M_{nn})\det{(M)_{(n-1)\times (n-1)}}  + \sum ({\cdot})\right],
\end{equation}
where $\sum ({\cdot})$ is the remaining terms in the Laplace expansion.
The second term is
\begin{equation}
- a_{21}\det[(M)_{n(n-1)}|r]_{n \times n}=-a_{21}a_{12}\det{(M)_{(n-1)\times(n-1)}}.
\end{equation}
Vanishing of determinant of $A$ means that $a_{21}a_{12} = a_{11}a_{22}$. Therefore we see that
eq.~(\ref{1}) simply gives
\begin{flalign}
\det{N}&=a_{22}\left[M_{nn}\det[(M)_{(n-1)\times (n-1)}]  + \sum (\cdot)\right] \nonumber \\
&= a_{22}\det[M].
\end{flalign}
\end{proof}

It immediately follows that if $a_{22}\neq 0$, then $\det{N} \neq 0$ since $\det{M} \neq 0$.
This establishes the claim in section~\ref{characteristic matrix}, that the determinant for the $7 \times 7$ matrix is generically nonzero.

\end{document}